\documentclass[letterpaper, 10 pt, conference]{ieeeconf}

\IEEEoverridecommandlockouts 
\usepackage{cite}
\usepackage{amsmath,amssymb,amsfonts}
\usepackage{url}
\usepackage{color}
\usepackage{hyperref}
\usepackage{graphicx}
\usepackage{lipsum}  
\usepackage{optidef}
\usepackage{balance} 
\usepackage{enumerate}
\usepackage{mathtools}
  \usepackage{pgfplots}
  \pgfplotsset{compat=newest}
  \usetikzlibrary{plotmarks}
  \usetikzlibrary{arrows.meta}
  \usepgfplotslibrary{patchplots}
  \usepackage{grffile}
  \usepackage{amsmath}
\usepackage{apptools}
\usepackage{algorithm2e}
\RestyleAlgo{ruled}
\AtAppendix{\counterwithin{theorem}{subsection}}
\newtheorem{theorem}{Theorem}[section]

\newtheorem{remark}{Remark}

\newcommand{\map}[3]{#1: #2 \rightarrow #3}

\newcommand{\real}{\mathbb{R}}

\newcommand{\norm}[1]{\left\lVert#1\right\rVert}

\newcommand\oprocendsymbol{\hbox{$\square$}}
\newcommand\oprocend{\relax\ifmmode\else\unskip\hfill\fi\oprocendsymbol}

\newcommand*{\QEDA}{\hfill\ensuremath{\blacksquare}}

\DeclareMathOperator{\diag}{diag}

\DeclareMathOperator*{\tr}{tr}

\newcommand{\xeq}{x_{\rm eq}}

\newcommand{\Sigmaerri}{\Sigma_{{\rm err}}^{i}}

\newcommand{\Aext}{A_{\rm ext}}

\newcommand{\Aerri}{A_{{\rm err}}^i}

\newcommand{\PextB}{P_{\rm ext}^{\rm B}}

\newcommand{\QextB}{Q_{\rm ext}^{\rm B}}

\newcommand{\AextB}{A_{\rm ext}^{\rm B}}

\newcommand{\BextB}{B_{\rm ext}^{\rm B}}

\newcommand{\Bext}{B_{\rm ext}}

\newcommand{\Berri}{B_{{\rm err}}^i}

\newcommand{\CextB}{C_{\rm ext}^{\rm B}}

\newcommand{\Cext}{C_{\rm ext}}

\newcommand{\Cerri}{C_{{\rm err}}^i}

\def\0{{\mathbf{0}}}

\newcommand{\C}{\boldsymbol{C}}

\newcommand{\RN}{\mathbb{R}^N}

\newcommand{\Rp}{\mathbb{R}^p}

\newcommand{\Rm}{\mathbb{R}^m}
\newcommand{\RNN}{\mathbb{R}^{N\times N}}

\newcommand{\RNm}{\mathbb{R}^{N\times m}}

\newcommand{\RpN}{\mathbb{R}^{p\times N}}

\newcommand{\Rpm}{\mathbb{R}^{p\times m}}

\newcommand{\Saux}{\Sigma_{\rm aux}}
\newcommand{\Sext}{\Sigma_{{\rm ext}}}
\newcommand{\Serr}{\Sigma_{{\rm err}}}

\newcommand{\SextB}{\Sigma_{\rm ext}^{\rm B}}

\newcommand{\SiD}{\Sigma_i^{\rm D}}

\newcommand{\giD}{g_i^{\rm D}}
\newcommand{\SoneD}{\Sigma_1^{\rm D}}
\newcommand{\SoneL}{\Sigma_1^{\rm L}}
\newcommand{\StwoD}{\Sigma_2^{\rm D}}
\newcommand{\StwoL}{\Sigma_2^{\rm L}}
\newcommand{\SlD}{\Sigma_l^{\rm D}}
\newcommand{\SlL}{\Sigma_l^{\rm L}}
\newcommand{\SjL}{\Sigma_j^{\rm L}}
\newcommand{\SjD}{\Sigma_j^{\rm D}}
\newcommand{\SD}{\Sigma^{\rm D}}

\newcommand{\SiL}{\Sigma_i^{\rm L}}

\newcommand{\giL}{g_i^{\rm L}}

\newcommand{\SiLH}{\widehat{\Sigma}_i^{\rm L}}

\newcommand{\giLH}{\widehat{g}_i^{\rm L}}
\newcommand{\SL}{\Sigma^{\rm L}}

\newcommand{\Hnorm}{\mathcal{H}_2}

\newcommand{\HnormInf}{\mathcal{H}_{\infty}}

\newcommand{\Jcal}{\mathcal{J}}

\newcommand{\Jcali}{\mathcal{J}_i}

\graphicspath{{../../img/}}

\begin{document}
\title{\LARGE \bf Controllable Neural Architectures for Multi-Task Control}

\author{Umberto Casti, Giacomo Baggio, Sandro Zampieri, and Fabio
  Pasqualetti \thanks{This material is based upon work supported in
    part by awards AFOSR-FA9550-20-1-0140, and
    AFOSR-FA9550-19-1-0235. Fabio Pasqualetti are with the Department
    of Mechanical Engineering, University of California at Riverside,
    \href{mailto:fabiopas@ucr.edu}{\texttt{fabiopas@ucr.edu.}}
    Umberto Casti, Giacomo Baggio, and Sandro Zampieri are with the
    Department of Information Engineering, University of Padova, Italy
    \{\href{mailto:castiumber@dei.unipd.it}{\texttt{castiumber}},\href{mailto:baggio@dei.unipd.it}{\texttt{baggio}},
    \href{mailto:zampi@dei.unipd.it}{\texttt{zampi\}@dei.unipd.it}}}}
\maketitle


\begin{abstract} 
  This paper studies a multi-task control problem where multiple
  linear systems are to be regulated by a single non-linear
  controller. In particular, motivated by recent advances in
  multi-task learning and the design of brain-inspired architectures,
  we consider a neural controller with (smooth) ReLU activation
  function. The parameters of the controller are a connectivity matrix
  and a bias vector: although both parameters can be designed, the
  connectivity matrix is constant while the bias vector can be varied
  and is used to adapt the controller across different control
  tasks. The bias vector determines the equilibrium of the neural
  controller and, consequently, of its linearized dynamics. Our
  multi-task control strategy consists of designing the connectivity
  matrix and a set of bias vectors in a way that the linearized
  dynamics of the neural controller for the different bias vectors
  provide a good approximation of a set of desired controllers. We
  show that, by properly choosing the bias vector, the linearized
  dynamics of the neural controller can replicate the dynamics of any
  single, linear controller. Further, we design gradient-based
  algorithms to train the parameters of the neural controller, and we
  provide upper and lower bounds for the performance of our neural
  controller. Finally, we validate our results using different
  numerical examples.
\end{abstract}

\section{Introduction}\label{sec: introduction}
Control algorithms are typically tuned to optimize the performance of
a single dynamical system. Similarly, machine learning algorithms are
often trained for specific datasets and require time-consuming
retraining procedures to accomodate changes in the data and
objectives\cite{BJ:2000,RC:97}. On the other hand, many natural
systems can seamlessly adapt across different tasks and transfer
learned skills to new and unseen contexts. In the human brain, for
instance, astrocytes are believed to bias neuronal functioning to
provide contextual adaptation capabilities \cite{LG-FP-TP-SC:23}
without changing neuronal coupling. Motivated by the discrepancy
between natural and artificial systems and the need to alleviate
retraining times and requirements, techniques for multi-task learning
have recently been developed \cite{AM-SL:18,YY-XS-CH-JZ:19}, showing
that a single artificial architecture can in fact learn to solve
multiple tasks. Yet, techniques for multi-task control have remained
elusive.

In this paper, we propose a non-linear neural controller to solve a
multi-task control problem. We consider a controller inspired by
neural architectures \cite{PJW:89} with (smooth) ReLU activation
function (see Fig. \ref{fig:figExample}). The parameters of the
controller are the states connectivity matrix, whose value is trained
at design time and remains constant, and a bias vector, whose value
depends on the control problem at hand and is selected among a set of
values trained at design time. Selecting the bias vector is a
convenient way to provide the controller with the ability to adapt to
different dynamical systems and tasks without the need to retrain the
states connectivity matrix. We emphasize that the main
  objective of this work is to validate the ability of our nonlinear
  controller to approximate the behavior of a set of desired linear
  controllers, rather than to solve any specific control
  problem. Loosely speaking, our approach takes inspiration from the
  human brain that, despite a relatively static neuronal network,
  modulates neuronal responses to accomodate contextual and task
  changes.
%
%
%

\begin{figure}
 \centering   
 \input{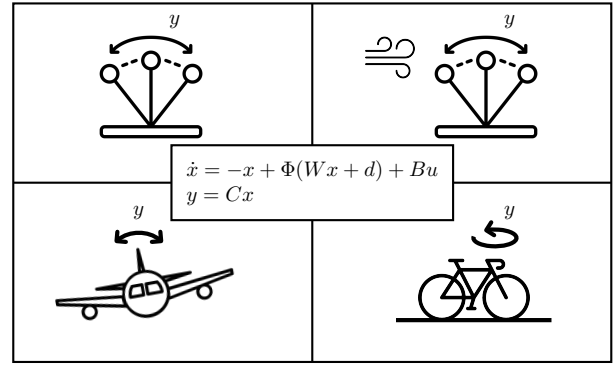} 
 \caption{An illustration of the multi-task control problem considered in
   this paper, where a set of linear systems is to be regulated by a
   single non-linear controller. This paper proposes a neural
   controller with (smooth) ReLU activation function and two
   parameters. The connectivity matrix $W$ is typically large and
   remains constant across the different tasks. The bias vector $d$,
   instead, is low-dimensional and is used to adapt the performance
   of the neural controller to different control tasks. See Section
   \ref{sec:prelim} for a detailed explanation of the neural
   controller and our multi-task control problem, and Section
   \ref{sec:grad} for a numerical study of this multi-task control
   example.}
    \label{fig:figExample}
\end{figure}

\smallskip
\noindent
\textbf{Related Work.} The literature on multi-task control is limited. Traditional controller design methods, such as the linear quadratic regulator and model predictive control~\cite{BDA-JBM:07,BK-MC:16}, are tailored for single dynamical systems and often require a complete redesign when system dynamics change.

Adaptive control is crucial for managing systems with significant
uncertainties, where robust techniques
fail~\cite{JPH-DL-ASM:03}. These frameworks typically employ a family
of controllers with parameters that vary
smoothly~\cite{KJA:91}. However, when system parameters affect
dynamics in complex ways, constructing a continuously parameterized
set of controllers becomes difficult, especially if high robustness
and performance are required. To mitigate these challenges, approaches
like logic-based switching strategies~\cite{BM:85,MF-BB:86,DEM-EJD:91}
have been proposed, focusing on discrete controller switching rather
than continuous adjustments. Our work, in contrast, studies the
approximation properties of a neural controller, which is independent
but could be integrated with such switching-based methods.

Recent studies in multi-task control, such as~\cite{LX-LY-GC-SS:22,YC-AMO-FP-EDA:23}, explore system identification across multiple datasets, while others~\cite{TTZ-KK-BDL:22,TG-AAAM-VK-FP:23} address transfer and imitation control. Differently
from these approaches, this paper considers a setting where
the system to be controlled varies abruptly and arbitrarily.

Multi-task learning methods enable neural networks to handle diverse problems, often employing techniques like masks to select task-specific network components~\cite{AM-SL:18}. However, these methods are prone to catastrophic forgetting, where models lose previously learned knowledge when exposed to new tasks~\cite{BR-DT-AZ:22,GIP-RK-JLP-CK-SW:19}. Despite the extensive literature on multi-task learning~\cite{AM-SL:18,SM-AO-LA-PJ:2014,YY-XS-CH-JZ:19,AAD-UD-CS:2017}, such techniques do not directly address control challenges or provide performance guarantees.

The work most similar to ours is~\cite{DK-MB-MFA:96}, which studies simultaneous approximation of multiple systems using a single approximating model. In contrast, our method allows for multiple approximating systems that share a common connectivity matrix but differ through low-dimensional bias vectors. This relationship between systems introduces complexity, making existing LMI techniques~\cite{DK-MB-MFA:96} inapplicable in a straightforward manner.

Finally, our architecture draws inspiration from neuroscience, particularly the interplay between astrocytes and neurons in the human brain~\cite{SGG-NS-LP-JFO:18,CHTT-GP-GRG:18,CMR-SC-TP:23}. Emerging theories suggest that astrocytes modulate neuronal function, enabling adaptive responses without altering the network structure~\cite{LG-FP-TP-SC:23}. Our neural controller mimics this biological mechanism, where the bias vector acts analogously to astrocytic modulation, facilitating context-dependent adaptation.

\smallskip
\noindent
\textbf{Paper contribution.} The main contributions of this paper are
as follows. First, we formulate a novel multi-task control problem,
where a set of known linear systems is to be regulated by a (possibly
varying) single controller. We propose a novel control strategy based
on a non-linear neural controller with (smooth) ReLU activation function and two
parameters: a connectivity matrix and a bias vector. While the
controller connectivity matrix is trained at design stage and remains
constant, the values of the low-dimensional bias vector are trained at
design stage but can vary over time. Changing the bias vector modifies
the equilibria of the neural controller and its linearized dynamics,
and allows the controller to approximate different desired linear
dynamics by tuning a small subset of the parameters. Second, we prove
that, by properly choosing the bias vector, the linearized dynamics of
the neural controller can replicate the dynamics of any linear
system. Third, we provide a gradient-based algorithm to train the
parameters of the neural controllers in a way that its linearized
dynamics obtained by appropriately changing the bias vector
approximate a set of desired linear dynamics. Fourth and finally, we
provide upper and lower bounds on the performance of our multi-task control
problem. While some bounds are of technical nature, others show of the
approximation capabilities of the neural controller depend on the
dimension, number and similarity of the desired linear dynamics, and
the dimension of the neural controller.

\smallskip
\noindent
\textbf{Paper organization.} The rest of the paper is organized as
follows. Section \ref{sec:prelim} contains our problem formulation and
preliminary results. Section \ref{sec:grad} contains our numerical
algorithms and some numerical results. Finally, Section
\ref{sec:Bounds} contains our lower and upper bounds on the multi-task
control problem, and Section \ref{sec: conclusion} concludes the
paper.

\section{Problem setup and preliminary notions}\label{sec:prelim}
Consider the following non-linear neural controller:
\begin{equation}\label{eq:appSys}
\Sigma = \left\lbrace\begin{array}{@{}l@{\,}l}
\dot{x} &= -x + \Phi \left( W x + d \right) + B u, \\
y &= C x ,
\end{array}\right. 
\end{equation}
where $x \in \RN$, $d \in \RN$, $u \in \Rm$, and
$y \in \Rp$ are the state, a free parameter, input, and
output of the controller, respectively, and $W \in \RNN$, $B \in \RNm$
and $C \in \RpN$ are the controller matrices. The activation
function $\map{\Phi}{\RN}{\left(0,\,+\infty\right)^N}$ is the elementwise application of the
(smooth) ReLU function $\map{\phi}{\real}{\left(0,\,+\infty\right)}$, which is
defined as $\phi\left(x\right) = \ln\left(1+e^{x}\right)$. Further,
when $u = 0$, the equilibria of the neural controller
\eqref{eq:appSys} satisfy the equation
\begin{equation}\label{eq:eqPoints}
  \xeq = \Phi \left( W \xeq + d \right),
\end{equation}
and, locally, obey the linearized dynamics
\begin{equation}\label{eq:linSys}
  \SL=\left(-I +  \underbrace{\diag\left(\Phi_\text{d} (W\xeq + d)\right)}_{D} W,\,B,\,C\right),
\end{equation}
where $\map{\Phi_\text{d}}{\real^N}{\left(0,\,1\right)^{N}}$ returns the elementwise application of the function
\begin{equation*}
   \phi_{\text{d}} \left( x\right) = \frac{d}{dx}\phi\left(x\right) =
  \frac{1}{1+e^{-x}} .
\end{equation*}

\begin{theorem}{\textbf{\emph{(Parametrization using
        $d$)}}}\label{thm:dD}
  For any matrix $W \in \real^{N\times N}$ and vector
  $\bar{d} \in \left(0,\,1\right)^N$, there exists $\xeq$ and $d$ that
  satisfy equation \eqref{eq:eqPoints} and
  $\Phi_\text{d} (W\xeq + d) = \bar d$.
\end{theorem}
\begin{proof}
  Notice that $\phi_{\text{d}}$ is a injective function, so
  that its inverse is well defined. Let
  $\xeq = \Phi\left[\Phi_\text{d}^{-1}\left(\bar{d}\,\right)\right]$ and
  $d =
  \Phi_\text{d}^{-1}\left(\bar{d}\,\right)-W\Phi\left[\Phi_\text{d}^{-1}\left(\bar{d}\,\right)\right]$. Then,
  from \eqref{eq:eqPoints},
  \begin{align*}
    \xeq &= \Phi \left( W
    \Phi\left[\Phi_\text{d}^{-1}\left(\bar{d}\,\right)\right] +
    \Phi_\text{d}^{-1}\left(\bar{d}\right)-W\Phi\left[\Phi_\text{d}^{-1}\left(\bar{d}\,\right)\right]
    \right)\\
    &= \Phi\left( \Phi_\text{d}^{-1}\left(\bar{d}\,\right) \right) =
      \xeq .
  \end{align*}
  Further, to conclude,
  \begin{align*}
    \Phi_\text{d} \left( W
    \Phi\left[\Phi_\text{d}^{-1}\left(\bar{d}\,\right)\right] +
    \Phi_\text{d}^{-1}\left(\bar{d}\,\right)-W\Phi\left[\Phi_\text{d}^{-1}\left(\bar{d}\,\right)\right]
    \right) = \bar d.
  \end{align*}
\end{proof}
Theorem \ref{thm:dD} shows that there exists a vector $d$ and an
equilibrium $\xeq$ that realizes any desired matrix $D$ in
\eqref{eq:linSys}. Hence, in what follows, we derive conditions and
algorithms for the matrix $D$, with the understanding that such matrix
can ultimately be realized by choosing the vector $d$.
Now we are ready to formally state our multi-task control  problem, consider a set of $M$ distinct, stable, controllable, and
observable linear, time-invariant systems denoted as
\begin{equation}\label{eq:desSystems}
  \SiD = \left(A_i,\,B_i,\,C_i\right) ,
\end{equation}
with $A_i \in \real^{n \times n}$, $B_i \in \real^{n \times m}$,
$C_i \in \real^{p \times n}$ and $i = 1,\dots, M$. Our multi-task control problem is
\begin{equation}\label{eq:problem2}
\min_{\substack{W,D_1, \dots,D_M,\\ \,B ,\, C}}
\;\; \sum_{i=1}^{M}\norm{\SiD-\SiL }^2_{2} ,
\end{equation}
where $\SiL$ denotes the $i$-th linearized dynamics \eqref{eq:linSys}
with diagonal matrix $D_i$ and $\norm{\cdot}_{2}$ the $\Hnorm$-norm.
\begin{remark}
  Problem~\eqref{eq:problem2} is treated as an optimal approximation
  problem rather than a traditional control problem, since the focus
  of this paper is on the approximation capabilities of the
  biologically inspired non-linear neural
  controller~\eqref{eq:appSys}, rather than the control performance of
  the systems $\SiD$, which may represent LQR controllers or
  general~systems. We also note that the implementation of a neural
  controller may be more efficient than the separate implementation of
  multiple controllers. In fact, storing $M$ distinct linear
  controllers as in \eqref{eq:desSystems} requires $Mn^4mp$ parameters
  that define the matrices $A_i$, $B_i$, and $C_i$. In contrast, the
  neural controller~\eqref{eq:appSys} only needs $MN +N^4mp$
  parameters, making it more efficient as $M$ increases and
  $N < n^4mp$.
  \oprocend

\end{remark}
\begin{remark}
  Although Problem~\eqref{eq:problem2} is formulated under the
  assumption that all systems in the set~\eqref{eq:desSystems} share
  the same state dimension $n$, this simplification is made primarily
  for the ease of notation. The theoretical derivations, including the
  gradient computations discussed in Section~\ref{sec:grad}, are
  easily adaptable to scenarios where the systems $\SiD$ have distinct
  state dimensions $n_i$.
  \oprocend
\end{remark}

In the minimization problem
\eqref{eq:problem2}, the optimization variables allow the neural
controller \eqref{eq:appSys} to approximate the desired systems
$\eqref{eq:desSystems}$ locally around its equilibrium points. The
approximation error in \eqref{eq:problem2} depends in a nontrivial way
on several parameters, including the dimension of the neural
controller, the number, and the
diversity of the systems to be approximated. In the following
sections, we define both upper and lower bounds on the approximation
error, as detailed in Section~\ref{sec:Bounds}. Furthermore, in
Section~\ref{sec:grad}, we derive the gradient useful to implement a
numerical procedure based on a gradient descent to solve the
minimization problem~\eqref{eq:problem2}.

\section{Gradient-based multi-task control}\label{sec:grad}
This section contains the analytical expression of the gradient of the
multi-task control problem \eqref{eq:problem2} with respect to
the matrices of the neural controller. These expressions can be used
to numerically optimize the performance of the neural controller given
a set of desired control tasks. To this aim, define the following
error system and matrices:
\begin{align*}
\Sigmaerri = (\Aerri, \Berri, \Cerri),
\end{align*} with
\begin{equation*}
\begin{split}
    \Aerri &=
    \begin{bmatrix}
      A_i & 0\\
      0 & -I+D_i W\\
    \end{bmatrix},\; 
    \Berri =
    \begin{bmatrix}
      B_i\\
      B
    \end{bmatrix}, \\
    \Cerri &=
    \begin{bmatrix}
      C_i & -C
    \end{bmatrix} ,
  \end{split}
\end{equation*}
and observability ($Q^i$) and controllability ($P^i$) Gramians as
\begin{equation*}
Q^i = \left[\begin{array}{cc}
Q^i_{11}&Q^i_{12}\\
{Q^i_{12}}^{\top}&Q^i_{22}
            \end{array}\right]
          ,\;
          P^i =\left[\begin{array}{cc}
P^i_{11}&P^i_{12}\\
{P^i_{12}}^{\top}&P^i_{22}
\end{array}\right].
\end{equation*}

\begin{theorem}{\textbf{\emph{(Analytical gradient of
        \eqref{eq:problem2})}}}\label{thm: gradient}
  The gradient of the minimization problem \eqref{eq:problem2} is as
  follows:
  \begin{align}\label{eq:gradTab}
    \begin{split}
      \frac{\partial \Jcal}{\partial W} &= 2\sum_{i=1}^{M}
      D_i\left({Q^i_{12}}^{\top}{P^i_{12}}+{Q^i_{22}}{P^i_{22}}\right),
      \\
      \frac{\partial \Jcal}{\partial B}
      &=2\sum_{i=1}^{M}\left({Q^i_{12}}^{\top}
        B_i+{Q^i_{22}} B\right), \\
      \frac{\partial \Jcal}{\partial C} &=2\sum_{i=1}^{M}\left( -
        C_i{P^i_{12}+C{P^i_{22}}}\right),
      \\
      \frac{\partial\Jcal}{\partial\left(D_i\right)_{hh}} &=
      2\left(\left({Q^i_{12}}^{\top}{P^i_{12}}
          + {Q^i_{22}}
          {P^i_{22}}
        \right)W^{\top}\right)_{hh}
      .
    \end{split}
  \end{align}
\end{theorem}
\medskip
\begin{proof}
  Let
  \begin{equation*}
    \Jcali = \norm{\SiD-\SiL }^2_{2} \quad i = 1,\ldots,\,M ,
  \end{equation*}
  and notice from \cite{JV-BV-WM-SV-MD:09} that 
  \begin{equation*}
    \frac{\partial \Jcali}{\partial\Aerri} = \frac{\partial
      \norm{\Sigmaerri}_2^2}{\partial\Aerri} = 2Q^i P^i .
  \end{equation*}
  Using the chain rule \cite{KBP-MSP:12} we obtain
  \begin{equation}\label{eq:gramW}
    \begin{aligned}
      \frac{\partial \Jcali}{\partial \left(W\right)_{hl}} &= \sum_{k=1}^{N+n}\sum_{j=1}^{N+n}\frac{\partial \Jcali}{\partial\left(\Aerri\right)_{kj}}\frac{\partial \left(\Aerri\right)_{kj}}{\partial \left(W\right)_{hl}} \\&= \frac{\partial \Jcali}{\partial\left(\Aerri\right)_{\left(n+h\right)\left(n+l\right)}}\frac{\partial \left(\Aerri\right)_{\left(n+h\right)\left(n+l\right)}}{\partial \left(W\right)_{hl}}\\
      &= 2\left(D_i\right)_{hh}\left( Q_i
        P_i\right)_{\left(n+h\right)\left(n+l\right)}, 
    \end{aligned}
  \end{equation}
  where we have used the fact that, for any $1 \leq h\leq N$ and $1 \leq
  l\leq N$, it holds
  \begin{equation*}
    \frac{\partial\left(\Aerri\right)_{kj}}{\partial w_{hl}} =
    \left\lbrace
      \begin{array}{@{}c@{\quad}l}
        \left(D_i\right)_{hh} & \text{if } \left(k,\,j\right)=\left(n+h,\,n+l\right)\\
        0
                              & \text{otherwise} 
      \end{array}
    \right..
  \end{equation*}
  Rewriting~\eqref{eq:gramW} in compact matrix form and summing over
  the index $i$ we obtain the first equation in
  \eqref{eq:gradTab}. With a similar reasoning we obtain
  \begin{align*}
    \frac{\partial\left(\Aerri\right)_{kj}}{\partial\left(D_i\right)_{hh}}
    = \left\lbrace
    \begin{array}{@{}c@{\quad}l}
      w_{h\left(j-n\right)} & \text{if } k=n+h \land n< j\leq n+N\\
      0 & \text{otherwise}
    \end{array}\right.
  \end{align*}
  and, for any $i=1,\ldots,\,M$,
\begin{equation*}
\frac{\partial\Jcal}{\partial\left(D_i\right)_{hh}}=\frac{\partial\Jcali}{\partial\left(D_i\right)_{hh}} = 2\left(Q^i P^i \begin{bmatrix}0&0\\0& W^{\top}\end{bmatrix}\right)_{\left(n+h\right)\left(n+h\right)}.
\end{equation*}
This leads to the last equation of~\eqref{eq:gradTab}. 

The derivation of $\frac{\partial \Jcal}{\partial C}$ and
$\frac{\partial \Jcal}{\partial B}$ is different. Recall that
\begin{equation*}
\begin{aligned}
  \norm{\Sigmaerri}^2_{2} &= \tr \left(\Cerri P
    ^i{\Cerri}^{\top}\right) = \tr
  \left({\Berri}^{\top}Q^i\Berri\right).
\end{aligned}
\end{equation*}
Notice that $\frac{\partial \Jcali}{\partial\Berri} = 2Q^i\Berri$,
and that 
\begin{align*}
\frac{\partial \Jcal i}{\partial \left(B\right)_{hl}} &=
                                           \sum_{k=1}^{N+n}\sum_{j=1}^{m}\frac{\partial
                                           \Jcali}{\partial\left(\Berri\right)_{kj}}\frac{\partial
                                           \left(\Berri\right)_{kj}}{\partial
                                           \left(B\right)_{hl}} \\&= \frac{\partial
  \Jcali}{\partial\left(\Berri\right)_{\left(n+h\right)l}} =
  2\left(Q^i\Berri\right)_{\left(n+h\right)l}.
\end{align*}
Equivalently in compact form, we have
\begin{equation}\label{eq:gradB}
  \frac{\partial \Jcali}{\partial B} = 2\left({Q^i_{12}}^{\top}
    B_i+{Q^i_{22}} B\right) ,
\end{equation}
and, with a similar procedure,
\begin{equation}\label{eq:gradC}
\frac{\partial \Jcali}{\partial C} = 2\left( - C_i{P^i_{12}} + C{P^i_{22}}\right).
\end{equation}
To conclude, summing over $i = 1,\ldots,M$ on~\eqref{eq:gradB}
and~\eqref{eq:gradC} we get the remaining two equations
of~\eqref{eq:gradTab}.
\end{proof}

The gradient in Theorem \eqref{thm: gradient} allows us to use
gradient descent methods to approximate a solution to the minimization
problem \eqref{eq:problem2}.
We conclude this section with an example.

Consider the systems in Fig. \ref{fig:figExample} with simplified dynamics
\begin{align}\label{eq: systems}
  \begin{split}
\Sigma_\text{Aircraft} &= \left(\begin{bmatrix}
0 & 0\\
0&0
\end{bmatrix}, \begin{bmatrix}
 0\\
\frac{r}{J}
\end{bmatrix} ,\begin{bmatrix}
 0&1
\end{bmatrix}\right), \\
\Sigma_\text{Pendulum} &= \left(\begin{bmatrix}
0 & 1\\
\frac{m g l}{J_t}&0
\end{bmatrix}, \begin{bmatrix}
 0\\
\frac{1}{J_t}
\end{bmatrix},\begin{bmatrix}
 1&0
\end{bmatrix}\right) , \\
\Sigma_\text{Pendulum + frict.} &= \left(\begin{bmatrix}
0 & 1\\
\frac{m g l}{J_t}&\frac{\gamma}{J_t}
\end{bmatrix} , \begin{bmatrix}
 0\\
\frac{1}{J_t}
\end{bmatrix} ,\begin{bmatrix}
 1&0
\end{bmatrix}\right),\\
\Sigma_\text{Bicycle} &= \left(\begin{bmatrix}
0 & 1\\
\frac{m g h}{J}&0
\end{bmatrix} , \begin{bmatrix}
 \frac{Dv_0}{bJ}\\
\frac{m v_0^2 h}{bJ}
\end{bmatrix} ,\begin{bmatrix}
 1&0
\end{bmatrix}\right).
\end{split}
\end{align}

Let $g = 9.8$ and other specific parameters as in \cite{KJA-RMM:08}:
\begin{gather*}
\begin{array}{cccccc}
\hline
J_p & m & l & \gamma & J&r\\
\hline
0.006 & 0.2 & 0.3 & 0.01&0.0475&25\\
\hline
\end{array}\\
\begin{array}{cccccc}
\hline
D & v_0 & M & h &b & J_b \\
\hline
4.8 & 2.5 & 8 & 1 & 1.2 & 8\\
\hline
\end{array}.
\end{gather*}
The systems $\SiD$, with $i=1,\ldots,4$, to be approximated with the
neural controller, of dimension $N=3$, are the Linear Quadratic Regulators of the systems
\eqref{eq: systems}, with identity cost
matrices. Fig. \ref{fig:impResponses} shows the impulse responses of
the desired systems $\SiD$ and the neural controller, which is
optimized using the gradient in Theorem \ref{thm: gradient}. While this
numerical gradient-based procedure offers no stability or performance
guarantees (see Section \ref{sec:Bounds} for some fundamental
performance limitations of our approach), our numerical studies show
promising results and demonstrate the viability of our multi-task control
approach.
\begin{figure}
  \centering
  \input{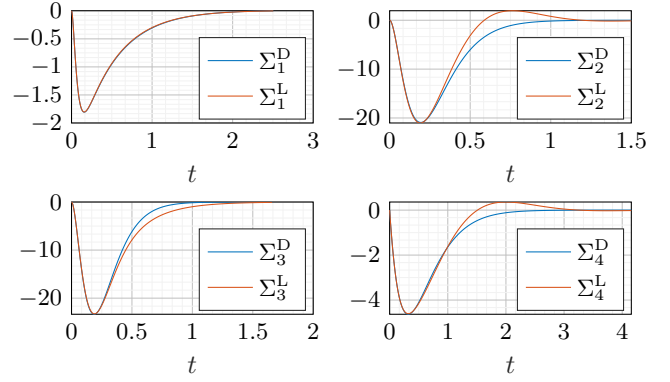}
  \caption{This figure displays the impulse responses of the
    feedback interconnection between the desired controllers $\SiD$ (or
    their implementations via the neural controller~\eqref{eq:appSys}) and the plants described in
    Fig.~\ref{fig:figExample}. The closed-loop impulse responses
    associated with the approximating controllers, which are the result of
    linearizing our neural controller, are highlighted in red. In
    contrast, the impulse responses of the closed-loop systems governed
    by the desired controllers are illustrated in
    blue.}\label{fig:impResponses}
\end{figure}%
An additional observation supporting this fact comes from the simulation results depicted in Fig.~\ref{fig:figSim5}. Here, we consider $M=5$ randomly generated SISO systems $(n=2)$ and we show the evolution of the
cost~\eqref{eq:problem2} as the dimension $N$ of the neural controller
increases. Notably, even for small values of $N$, the neural controller exhibits significant approximation capabilities.

In our final numerical example, we explore a scenario involving randomly generated SISO systems ($n=2$) and neural controllers of dimension $N=4$ . Fig.~\ref{fig:figSim4} demonstrates how the cost~\eqref{eq:problem2} evolves as the number of systems increases.
\begin{figure}
  \centering
%
%
\definecolor{mycolor1}{rgb}{0.00000,0.44700,0.74100}%
\definecolor{mycolor2}{rgb}{0.85000,0.32500,0.09800}%
\begin{tikzpicture}

\begin{axis}[
width=0.8\columnwidth ,
height=1.0in,
scale only axis,
xmin=0.5,
xmax=8.5,
xtick={1,2,3,4,5,6,7,8},
xlabel style={font=\color{white!15!black}},
xlabel={$N$ (dimension of neural controller)},
ymode=log,
ymin=1e-04,
ymax=4700,
yminorticks=true,
ymajorgrids,
ytick distance=10^2,
ylabel style={font=\color{white!15!black}},
ylabel={$\Jcal$},
axis background/.style={fill=white},
font=\small
]
\addplot [color=black, dashed, forget plot]
  table[row sep=crcr]{%
1	983.287096373297\\
1	1422.7973959941\\
};
\addplot [color=black, dashed, forget plot]
  table[row sep=crcr]{%
2	654.206132574768\\
2	685.687434754891\\
};
\addplot [color=black, dashed, forget plot]
  table[row sep=crcr]{%
3	10.107047393407\\
3	11.3426798796911\\
};
\addplot [color=black, dashed, forget plot]
  table[row sep=crcr]{%
4	1.12390457797941\\
4	1.38646180778787\\
};
\addplot [color=black, dashed, forget plot]
  table[row sep=crcr]{%
5	0.675901753441095\\
5	1.37867991579266\\
};
\addplot [color=black, dashed, forget plot]
  table[row sep=crcr]{%
6	0.0581744283438782\\
6	0.107016975026899\\
};
\addplot [color=black, dashed, forget plot]
  table[row sep=crcr]{%
7	0.0258909534693177\\
7	0.052373297343444\\
};
\addplot [color=black, dashed, forget plot]
  table[row sep=crcr]{%
8	0.0107216979483102\\
8	0.0178541870938623\\
};
\addplot [color=black, dashed, forget plot]
  table[row sep=crcr]{%
1	9.5828166359754\\
1	36.134682154245\\
};
\addplot [color=black, dashed, forget plot]
  table[row sep=crcr]{%
2	0.499628473879987\\
2	6.33112110906495\\
};
\addplot [color=black, dashed, forget plot]
  table[row sep=crcr]{%
3	0.431850020000116\\
3	1.32936578977201\\
};
\addplot [color=black, dashed, forget plot]
  table[row sep=crcr]{%
4	0.0218991280567537\\
4	0.177300886028807\\
};
\addplot [color=black, dashed, forget plot]
  table[row sep=crcr]{%
5	0.00599167409707444\\
5	0.0368919744416013\\
};
\addplot [color=black, dashed, forget plot]
  table[row sep=crcr]{%
6	0.00402943077542443\\
6	0.00837752681398804\\
};
\addplot [color=black, dashed, forget plot]
  table[row sep=crcr]{%
7	0.00169525455145163\\
7	0.00578898114750207\\
};
\addplot [color=black, dashed, forget plot]
  table[row sep=crcr]{%
8	0.000355112832764759\\
8	0.00426322871783359\\
};
\addplot [color=black, forget plot]
  table[row sep=crcr]{%
0.875	1422.7973959941\\
1.125	1422.7973959941\\
};
\addplot [color=black, forget plot]
  table[row sep=crcr]{%
1.875	685.687434754891\\
2.125	685.687434754891\\
};
\addplot [color=black, forget plot]
  table[row sep=crcr]{%
2.875	11.3426798796911\\
3.125	11.3426798796911\\
};
\addplot [color=black, forget plot]
  table[row sep=crcr]{%
3.875	1.38646180778787\\
4.125	1.38646180778787\\
};
\addplot [color=black, forget plot]
  table[row sep=crcr]{%
4.875	1.37867991579266\\
5.125	1.37867991579266\\
};
\addplot [color=black, forget plot]
  table[row sep=crcr]{%
5.875	0.107016975026899\\
6.125	0.107016975026899\\
};
\addplot [color=black, forget plot]
  table[row sep=crcr]{%
6.875	0.052373297343444\\
7.125	0.052373297343444\\
};
\addplot [color=black, forget plot]
  table[row sep=crcr]{%
7.875	0.0178541870938623\\
8.125	0.0178541870938623\\
};
\addplot [color=black, forget plot]
  table[row sep=crcr]{%
0.875	9.5828166359754\\
1.125	9.5828166359754\\
};
\addplot [color=black, forget plot]
  table[row sep=crcr]{%
1.875	0.499628473879987\\
2.125	0.499628473879987\\
};
\addplot [color=black, forget plot]
  table[row sep=crcr]{%
2.875	0.431850020000116\\
3.125	0.431850020000116\\
};
\addplot [color=black, forget plot]
  table[row sep=crcr]{%
3.875	0.0218991280567537\\
4.125	0.0218991280567537\\
};
\addplot [color=black, forget plot]
  table[row sep=crcr]{%
4.875	0.00599167409707444\\
5.125	0.00599167409707444\\
};
\addplot [color=black, forget plot]
  table[row sep=crcr]{%
5.875	0.00402943077542443\\
6.125	0.00402943077542443\\
};
\addplot [color=black, forget plot]
  table[row sep=crcr]{%
6.875	0.00169525455145163\\
7.125	0.00169525455145163\\
};
\addplot [color=black, forget plot]
  table[row sep=crcr]{%
7.875	0.000355112832764759\\
8.125	0.000355112832764759\\
};
\addplot [color=mycolor1, forget plot,fill=white]
  table[row sep=crcr]{%
0.75	36.134682154245\\
0.75	983.287096373297\\
1.25	983.287096373297\\
1.25	36.134682154245\\
0.75	36.134682154245\\
};
\addplot [color=mycolor1, forget plot,fill=white]
  table[row sep=crcr]{%
1.75	6.33112110906495\\
1.75	654.206132574768\\
2.25	654.206132574768\\
2.25	6.33112110906495\\
1.75	6.33112110906495\\
};
\addplot [color=mycolor1, forget plot,fill=white]
  table[row sep=crcr]{%
2.75	1.32936578977201\\
2.75	10.107047393407\\
3.25	10.107047393407\\
3.25	1.32936578977201\\
2.75	1.32936578977201\\
};
\addplot [color=mycolor1, forget plot,fill=white]
  table[row sep=crcr]{%
3.75	0.177300886028807\\
3.75	1.12390457797941\\
4.25	1.12390457797941\\
4.25	0.177300886028807\\
3.75	0.177300886028807\\
};
\addplot [color=mycolor1, forget plot,fill=white]
  table[row sep=crcr]{%
4.75	0.0368919744416013\\
4.75	0.675901753441095\\
5.25	0.675901753441095\\
5.25	0.0368919744416013\\
4.75	0.0368919744416013\\
};
\addplot [color=mycolor1, forget plot,fill=white]
  table[row sep=crcr]{%
5.75	0.00837752681398804\\
5.75	0.0581744283438782\\
6.25	0.0581744283438782\\
6.25	0.00837752681398804\\
5.75	0.00837752681398804\\
};
\addplot [color=mycolor1, forget plot,fill=white]
  table[row sep=crcr]{%
6.75	0.00578898114750207\\
6.75	0.0258909534693177\\
7.25	0.0258909534693177\\
7.25	0.00578898114750207\\
6.75	0.00578898114750207\\
};
\addplot [color=mycolor1, forget plot,fill=white]
  table[row sep=crcr]{%
7.75	0.00426322871783359\\
7.75	0.0107216979483102\\
8.25	0.0107216979483102\\
8.25	0.00426322871783359\\
7.75	0.00426322871783359\\
};
\addplot [color=mycolor2, forget plot]
  table[row sep=crcr]{%
0.75	451.237959839594\\
1.25	451.237959839594\\
};
\addplot [color=mycolor2, forget plot]
  table[row sep=crcr]{%
1.75	62.714551715937\\
2.25	62.714551715937\\
};
\addplot [color=mycolor2, forget plot]
  table[row sep=crcr]{%
2.75	3.39832070399847\\
3.25	3.39832070399847\\
};
\addplot [color=mycolor2, forget plot]
  table[row sep=crcr]{%
3.75	0.290681003938079\\
4.25	0.290681003938079\\
};
\addplot [color=mycolor2, forget plot]
  table[row sep=crcr]{%
4.75	0.111315103769115\\
5.25	0.111315103769115\\
};
\addplot [color=mycolor2, forget plot]
  table[row sep=crcr]{%
5.75	0.0301751592043727\\
6.25	0.0301751592043727\\
};
\addplot [color=mycolor2, forget plot]
  table[row sep=crcr]{%
6.75	0.0134706552691037\\
7.25	0.0134706552691037\\
};
\addplot [color=mycolor2, forget plot]
  table[row sep=crcr]{%
7.75	0.00456917001463904\\
8.25	0.00456917001463904\\
};
\end{axis}

\end{tikzpicture}%
\caption{This figure illustrates a box plot summarizing the final cost~\eqref{eq:problem2} across $17$ simulations with distinct sets $(M = 5)$ of randomly generated systems $\SiD$, as the dimension $N$ of the neural controller increases from $1$ to $8$. The median of the simulations is marked by the central red line in each box, while the bottom and top edges of the box delineate the 25\textsuperscript{th} and 75\textsuperscript{th} percentiles, respectively. The whiskers extend to cover approximately $99.3\%$ of the data, indicating the range within which most observations lie.}\label{fig:figSim5}
\end{figure}
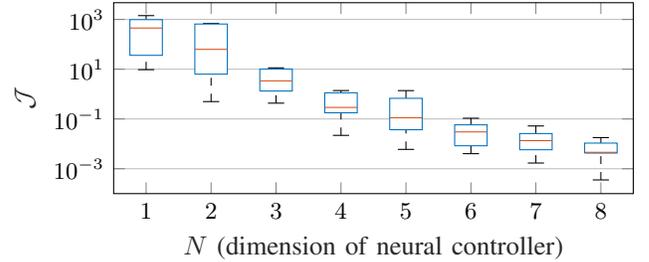
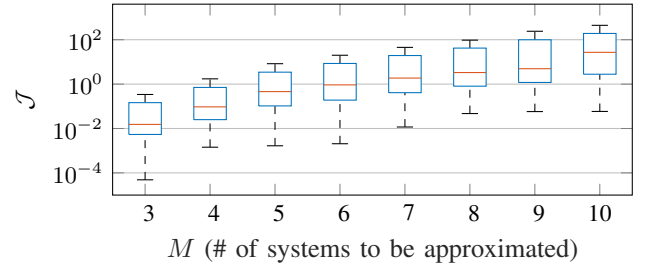
\begin{figure}
  \centering
%
%
\definecolor{mycolor1}{rgb}{0.00000,0.44700,0.74100}%
\definecolor{mycolor2}{rgb}{0.85000,0.32500,0.09800}%
\begin{tikzpicture}
\begin{axis}[%
width=0.8\columnwidth ,
height=1.0in,
scale only axis,
xmin=0.5,
xmax=8.5,
xtick={1,2,3,4,5,6,7,8},
xticklabels={{ 3},{ 4},{ 5},{ 6},{ 7},{ 8},{ 9},{10}},
xlabel style={font=\color{white!15!black}},
xlabel={$M$ (\# of systems to be approximated)},
ylabel={$\Jcal$},
ymode=log,
ymin=1e-05,
ymax=4000,
yminorticks=true,
ymajorgrids,
ytick distance=10^2,
axis background/.style={fill=white},
font=\small
]
\addplot [color=black, dashed, forget plot]
  table[row sep=crcr]{%
1	0.146072651905592\\
1	0.341170985297599\\
};
\addplot [color=black, dashed, forget plot]
  table[row sep=crcr]{%
2	0.708868244408245\\
2	1.72760403503667\\
};
\addplot [color=black, dashed, forget plot]
  table[row sep=crcr]{%
3	3.45514666043759\\
3	8.31306807334345\\
};
\addplot [color=black, dashed, forget plot]
  table[row sep=crcr]{%
4	8.52062257741646\\
4	19.9749320388551\\
};
\addplot [color=black, dashed, forget plot]
  table[row sep=crcr]{%
5	19.386369218122\\
5	45.2119792609884\\
};
\addplot [color=black, dashed, forget plot]
  table[row sep=crcr]{%
6	41.9854544944969\\
6	95.7783989951948\\
};
\addplot [color=black, dashed, forget plot]
  table[row sep=crcr]{%
7	99.417094397227\\
7	240.799127777064\\
};
\addplot [color=black, dashed, forget plot]
  table[row sep=crcr]{%
8	191.772702328322\\
8	445.512440156252\\
};
\addplot [color=black, dashed, forget plot]
  table[row sep=crcr]{%
1	4.87126442038237e-05\\
1	0.0053940079994554\\
};
\addplot [color=black, dashed, forget plot]
  table[row sep=crcr]{%
2	0.00143719107376088\\
2	0.0250668044977884\\
};
\addplot [color=black, dashed, forget plot]
  table[row sep=crcr]{%
3	0.00166411003184291\\
3	0.104128024103513\\
};
\addplot [color=black, dashed, forget plot]
  table[row sep=crcr]{%
4	0.00205799943901607\\
4	0.190019875163099\\
};
\addplot [color=black, dashed, forget plot]
  table[row sep=crcr]{%
5	0.0116779481338696\\
5	0.410812161320697\\
};
\addplot [color=black, dashed, forget plot]
  table[row sep=crcr]{%
6	0.0471030997397651\\
6	0.807327539284634\\
};
\addplot [color=black, dashed, forget plot]
  table[row sep=crcr]{%
7	0.0578567309674704\\
7	1.18747611927037\\
};
\addplot [color=black, dashed, forget plot]
  table[row sep=crcr]{%
8	0.058737743687953\\
8	2.79138609728855\\
};
\addplot [color=black, forget plot]
  table[row sep=crcr]{%
0.875	0.341170985297599\\
1.125	0.341170985297599\\
};
\addplot [color=black, forget plot]
  table[row sep=crcr]{%
1.875	1.72760403503667\\
2.125	1.72760403503667\\
};
\addplot [color=black, forget plot]
  table[row sep=crcr]{%
2.875	8.31306807334345\\
3.125	8.31306807334345\\
};
\addplot [color=black, forget plot]
  table[row sep=crcr]{%
3.875	19.9749320388551\\
4.125	19.9749320388551\\
};
\addplot [color=black, forget plot]
  table[row sep=crcr]{%
4.875	45.2119792609884\\
5.125	45.2119792609884\\
};
\addplot [color=black, forget plot]
  table[row sep=crcr]{%
5.875	95.7783989951948\\
6.125	95.7783989951948\\
};
\addplot [color=black, forget plot]
  table[row sep=crcr]{%
6.875	240.799127777064\\
7.125	240.799127777064\\
};
\addplot [color=black, forget plot]
  table[row sep=crcr]{%
7.875	445.512440156252\\
8.125	445.512440156252\\
};
\addplot [color=black, forget plot]
  table[row sep=crcr]{%
0.875	4.87126442038237e-05\\
1.125	4.87126442038237e-05\\
};
\addplot [color=black, forget plot]
  table[row sep=crcr]{%
1.875	0.00143719107376088\\
2.125	0.00143719107376088\\
};
\addplot [color=black, forget plot]
  table[row sep=crcr]{%
2.875	0.00166411003184291\\
3.125	0.00166411003184291\\
};
\addplot [color=black, forget plot]
  table[row sep=crcr]{%
3.875	0.00205799943901607\\
4.125	0.00205799943901607\\
};
\addplot [color=black, forget plot]
  table[row sep=crcr]{%
4.875	0.0116779481338696\\
5.125	0.0116779481338696\\
};
\addplot [color=black, forget plot]
  table[row sep=crcr]{%
5.875	0.0471030997397651\\
6.125	0.0471030997397651\\
};
\addplot [color=black, forget plot]
  table[row sep=crcr]{%
6.875	0.0578567309674704\\
7.125	0.0578567309674704\\
};
\addplot [color=black, forget plot]
  table[row sep=crcr]{%
7.875	0.058737743687953\\
8.125	0.058737743687953\\
};
\addplot [color=mycolor1, forget plot,fill=white]
  table[row sep=crcr]{%
0.75	0.0053940079994554\\
0.75	0.146072651905592\\
1.25	0.146072651905592\\
1.25	0.0053940079994554\\
0.75	0.0053940079994554\\
};
\addplot [color=mycolor1, forget plot,fill=white]
  table[row sep=crcr]{%
1.75	0.0250668044977884\\
1.75	0.708868244408245\\
2.25	0.708868244408245\\
2.25	0.0250668044977884\\
1.75	0.0250668044977884\\
};
\addplot [color=mycolor1, forget plot,fill=white]
  table[row sep=crcr]{%
2.75	0.104128024103513\\
2.75	3.45514666043759\\
3.25	3.45514666043759\\
3.25	0.104128024103513\\
2.75	0.104128024103513\\
};
\addplot [color=mycolor1, forget plot,fill=white]
  table[row sep=crcr]{%
3.75	0.190019875163099\\
3.75	8.52062257741646\\
4.25	8.52062257741646\\
4.25	0.190019875163099\\
3.75	0.190019875163099\\
};
\addplot [color=mycolor1, forget plot,fill=white]
  table[row sep=crcr]{%
4.75	0.410812161320697\\
4.75	19.386369218122\\
5.25	19.386369218122\\
5.25	0.410812161320697\\
4.75	0.410812161320697\\
};
\addplot [color=mycolor1, forget plot,fill=white]
  table[row sep=crcr]{%
5.75	0.807327539284634\\
5.75	41.9854544944969\\
6.25	41.9854544944969\\
6.25	0.807327539284634\\
5.75	0.807327539284634\\
};
\addplot [color=mycolor1, forget plot,fill=white]
  table[row sep=crcr]{%
6.75	1.18747611927037\\
6.75	99.417094397227\\
7.25	99.417094397227\\
7.25	1.18747611927037\\
6.75	1.18747611927037\\
};
\addplot [color=mycolor1, forget plot,fill=white]
  table[row sep=crcr]{%
7.75	2.79138609728855\\
7.75	191.772702328322\\
8.25	191.772702328322\\
8.25	2.79138609728855\\
7.75	2.79138609728855\\
};
\addplot [color=mycolor2, forget plot]
  table[row sep=crcr]{%
0.75	0.0153248510206488\\
1.25	0.0153248510206488\\
};
\addplot [color=mycolor2, forget plot]
  table[row sep=crcr]{%
1.75	0.0944567004338981\\
2.25	0.0944567004338981\\
};
\addplot [color=mycolor2, forget plot]
  table[row sep=crcr]{%
2.75	0.458786194775914\\
3.25	0.458786194775914\\
};
\addplot [color=mycolor2, forget plot]
  table[row sep=crcr]{%
3.75	0.923218297231139\\
4.25	0.923218297231139\\
};
\addplot [color=mycolor2, forget plot]
  table[row sep=crcr]{%
4.75	1.87224151095024\\
5.25	1.87224151095024\\
};
\addplot [color=mycolor2, forget plot]
  table[row sep=crcr]{%
5.75	3.31771976787348\\
6.25	3.31771976787348\\
};
\addplot [color=mycolor2, forget plot]
  table[row sep=crcr]{%
6.75	4.93885663781861\\
7.25	4.93885663781861\\
};
\addplot [color=mycolor2, forget plot]
  table[row sep=crcr]{%
7.75	27.0431741967785\\
8.25	27.0431741967785\\
};
\end{axis}
\end{tikzpicture}%
  \caption{This figure shows a box plot (with the format introduced in Fig.~\ref{fig:figSim5}) aggregating the final cost~\eqref{eq:problem2} across $100$ simulations for $100$ randomly selected sets of $\SiD$. For each set, we calculate the cost~\eqref{eq:problem2}, considering increasing subsets of systems to be approximated, with the number of systems $M$ increasing from $3$ to $10$.}\label{fig:figSim4}
\end{figure}

\section{Bounds on multi-task control problem}\label{sec:Bounds}
In this section we establish upper and lower bounds for the
optimization problem~\eqref{eq:problem2}, as a function of the number
and properties of the systems to be approximated.

\subsection{Upper bound}
To derive an upper bound, we notice that the approximation error in
\eqref{eq:problem2} obtained when choosing $M$ different matrices
$D_i$ is certainly bounded above by the error incurred when such $D_i$
matrices are all equal to each other. That is, solving the following
minimization problem provides an upper bound on the solution to the
minimization problem~\eqref{eq:problem2}:
\begin{align}\label{eq: upper bound problem}
  \begin{split}
    \min_{W,D,B , C}
    \;\; \sum_{i=1}^{M}\norm{\SiD-\SL }^2_{2} .
  \end{split}
\end{align}
The minimization problem \eqref{eq: upper bound problem} is akin to a
model reduction problem to approximate a given set of systems.

We start by introducing the necessary notation and preliminary steps
to present our result. Define the parallel system
\begin{align}\label{eq:parReal}
  \Sext = (\Aext, \Bext, \Cext),
\end{align}
with 
\begin{align*}
  \begin{split}
    \Aext &=
    \begin{bmatrix}
      A_1 & & \\
      & \ddots &\\
      &&A_M
    \end{bmatrix},\; 
    \Bext =
    \begin{bmatrix}
      B_1\\
      \vdots\\
      B_M
    \end{bmatrix}, \\
    \Cext &=
    \begin{bmatrix}
      C_1 & & \\
      & \ddots &\\
      &&C_M
    \end{bmatrix} ,
  \end{split}
\end{align*}
and its balanced and minimal realization \cite{MST-IP:87}
\begin{align}\label{eq:minBalReal}
  \SextB = (\AextB, \BextB, \CextB),
\end{align}
with
\begin{align}\label{eq: balanced matrices}
  \begin{split}
    \AextB &=
    \begin{bmatrix}
      {\AextB}_{11}&{\AextB}_{12}\\
      {\AextB}_{21}&{\AextB}_{22}
    \end{bmatrix}, \;\BextB=\begin{bmatrix}
      {\BextB}_{1}\\
      {\BextB}_{2}
    \end{bmatrix}, \\
    \CextB &=
    \begin{bmatrix}
      {\CextB}_{1} & {\CextB}_{2}
    \end{bmatrix} .
  \end{split}
\end{align}
Notice that the dimension of $\SextB$ is potentially smaller than the
dimension of $\Sext$ since the latter may not be a minimal
realization. Let $R$ be the dimension of $\SextB$ and let $N$ the
dimension of the sub block ${\AextB}_{11}$ when $R>N$.\footnote{If $R \leq N$ we let $\AextB = {\AextB}_{11}$, ${\BextB}={\BextB}_{1} $ and $\CextB={\CextB}_{1}$.}
Then, the controllability Gramian $\PextB$ and
observability Gramian $\QextB$ of \eqref{eq:minBalReal} are diagonal
and equal to each other:
\begin{equation}\label{eq: gramian balanced}
  \PextB = \QextB =
  \begin{bmatrix}
    S_1 & \\
    & S_2
  \end{bmatrix} >0,
\end{equation}
with $S_1 = \diag\left(\sigma_1,\ldots,\sigma_N\right)$,
$S_2 = \diag\left(\sigma_{N+1},\ldots,\sigma_R\right)$, and $\sigma_i$
Hankel singular value of the system \eqref{eq:minBalReal}\cite{AAC-DCS:01}.
\begin{theorem}{\emph{\bfseries (Upper bound
      of~\eqref{eq:problem2})}}\label{thm: upper bound}
  Let $\SL_1,\dots,\SL_M$ and $\SD_1,\dots,\SD_M$ be the LTI systems
  defined in equations~\eqref{eq:linSys} and~\eqref{eq:desSystems},
  respectively. Then if $R>N$:
\begin{equation}\label{eq:upperBound}
  \min_{\substack{W,D_1, \dots,D_M,\\ \,B ,\, C}} \sum_{i=1}^{M}\norm{\SiD-\SiL
  }^2_{2}\leq \left(\Jcal^{\rm B} + \sqrt{\sigma_1}\norm{\Delta C}_F\right)^2,
\end{equation}
if $R\leq N$
\begin{equation}\label{eq:upperBound1}
  \min_{\substack{W,D_1, \dots,D_M,\\ \,B ,\, C}} \sum_{i=1}^{M}\norm{\SiD-\SiL
  }^2_{2}\leq
  \sigma_1\norm{\Delta C}^2_F ,
\end{equation}
and where, using the notation in \eqref{eq: balanced
  matrices},\footnote{We use $\norm{\cdot}_F$ and
  $\norm{\cdot}_{\infty}$ to denote the Frobenius and the $\HnormInf$
  norms.}
\begin{align*}
  \Jcal^{\rm B} =
                  \sqrt{\tr\left[{{\BextB}_{2}}^{\top}S_2{\BextB}_{2}\right]
                  + 2N\norm{\Sigma_{\rm aux}}_{\infty}},
\end{align*}
with $\Saux = (\AextB, B_{\rm aux}, C_{\rm aux})$,
\begin{align*}
  B_{\rm aux} &=
                \begin{bmatrix}
                  0\\
                  S_2{\AextB}_{21}
                \end{bmatrix},
  C_{\rm aux} =
                \begin{bmatrix}
    0 & {\AextB}_{12}S_2
  \end{bmatrix},
\end{align*}
and
\begin{align*}
  \Delta C = \left(I-\frac{1}{M}
  \begin{bmatrix}
    I\\
    \vdots \\
    I\\
  \end{bmatrix}
  \begin{bmatrix}
    I &
    \cdots &
    I
  \end{bmatrix} \right){\CextB}_{1}.
\end{align*}
\end{theorem}
\bigskip

Some comments are in order to fully appreciate the result in Theorem
\ref{thm: upper bound}. First, $R > N$ when the dimension of the
neural controller is smaller than the number of different modes to be
approximated (as found through the balanced realization of
\eqref{eq:parReal}). Similarly, $R \le N$ when the neural controller
is larger than the number of different modes of the systems to be
approximated. Second, the term $\Jcal^{\rm B}$ in
\eqref{eq:upperBound} depends on the Hankel singular values that the
neural controller is not able to approximate ($S_2$ in \eqref{eq:
  gramian balanced}). This term vanishes when the dimension of the
neural controller is sufficiently large to capture all the modes of
the systems to be approximated (as in \eqref{eq:upperBound1}). As
similar error is also done when using the balanced truncation
technique to obtain a reduced dynamical model\cite{AAC-DCS:01}.
Third, the error $\Delta C$ in \eqref{eq:upperBound} and
\eqref{eq:upperBound1} is due to the fact that the system
\eqref{eq:minBalReal} to be approximated has more outputs than the
neural controller. To minimize such discrepancy and compute an upper
bound on the approximation error, Theorem \ref{thm: upper bound} uses
the average of the rows of the desired output matrix (namely,
$\frac{1}{M}\begin{bmatrix} I & \cdots & I\end{bmatrix}{\CextB}_{1}$,
which minimizes the discrepancy of the output matrices as measured by
the Frobenius norm). In the special case when the systems to be
approximated are all equal to each other, such error vanishes as the
average of the output matrices equals the actual output
matrices. Similarly, this error becomes small when the output matrices
in the balanced realization of the systems to be approximated are
similar across the systems to be approximated. Thus, Theorem \ref{thm:
  upper bound} shows that the multi-task control approximation error
depends on (i) the order of the neural controller through
$\Jcal^{\rm B}$, which dictates the number of different modes that can
be approximated, and (ii) the similarity of the systems to be
approximated through $\Delta C$. We are now ready to formally prove
Theorem \ref{thm: upper bound}.

\begin{proof}
  We derive the proof separately for two cases $R>N$ and $R\leq N$.
  
$\left(R>N\right)$ We use the minimization problem \eqref{eq: upper bound problem} as
  an upper bound for the minimization problem \eqref{eq:problem2}, and
  we compute the solution to \eqref{eq: upper bound problem} by using
  the $N$ dominant modes of the balanced realization
  \eqref{eq:minBalReal}. In particular, select $W$, $D$, and $B$ in
  \eqref{eq: upper bound problem} such that $-I+DW = {\AextB}_{11}$
  and $B = {\BextB}_{1}$, and
  $C = \frac{1}{M}\begin{bmatrix} I & \cdots &
    I\end{bmatrix}{\CextB}_{1}$.\footnote{The output matrix $C$ of the neural controller have different dimensions than ${\CextB}_1$, which prevents us from implementing the balanced truncation method to find the controller that minimize the cost \eqref{eq:
      upper bound problem}.} Then, the cost in \eqref{eq: upper bound
    problem} becomes
  \begin{align*}
    \sum_{i=1}^{M}\norm{\SiD-\SL }^2_{2} = \| \Sext - \Sext^\text{L}
    \|_2^2 ,
  \end{align*}
  where
  \begin{align*}
    \Sext^\text{L} = \left(
    {\AextB}_{11}, {\BextB}_1 ,
    \underbrace{
    \begin{bmatrix}
      C \\ \vdots\\C
    \end{bmatrix}}_{\Cext^L}
    \right) .
  \end{align*}
  Let $\Delta C = {\CextB}_1 - \Cext^L$ and notice that
  \begin{align}\label{eq: bound sum}
    \| \Sext - \Sext^\text{L}
    \|_2^2 \le \left(\| \Sext - {\Sext}_1  \|_2 + \| \Serr  \|_2\right)^2 ,
  \end{align}
  where ${\Sext}_1 = \left(
    {\AextB}_{11}, {\BextB}_1 , {\CextB}_1
                \right)$ and
  \begin{align*}
    \Serr = \left(
    {\AextB}_{11}, {\BextB}_1 , \Delta C
            \right) .
  \end{align*}
  Notice that
  \begin{align*}
    \norm{\Serr}^2_{2}=\tr\left(S_1\Delta C^{\top}\Delta C\right)\leq \sigma_1\norm{\Delta C}^2_F,
  \end{align*}
  where $S_1$ and $\sigma_1$ are as in \eqref{eq: gramian
    balanced}. In summary, leveraging the upper bound on the balanced truncation presented in~\cite{AAC-DCS:01}, inequality~\eqref{eq: bound sum} yields
  \begin{align*}
    \| \Sext - \Sext^\text{L}
    \|_2^2 \le
    \left(\Jcal^{\rm B} + \sqrt{\sigma_1}\norm{\Delta C}_F\right)^2 .
  \end{align*}
$\left(R\leq N\right)$ The balance realization ~\eqref{eq: balanced matrices} is already of order less than $N$ and represents a minimal realization of realization~\eqref{eq:parReal}. For this reason choosing
  \begin{align*}
-I+DW &= \begin{bmatrix}
\AextB & 0\\
0& -I \\
\end{bmatrix} & B = \begin{bmatrix}
\BextB\\
0
\end{bmatrix} \\ C &= \begin{bmatrix}
    \frac{1}{M}\begin{bmatrix} I & \cdots &
    I\end{bmatrix}\CextB & 0
    \end{bmatrix},
\end{align*}
the first term in~\eqref{eq: bound sum} is zero and then we obtain~\eqref{eq:upperBound1}.
\end{proof}

To conclude, we provide an example to evaluate the upper bound in
Theorem \ref{thm: upper bound}.
Using the dataset utilized for Fig.~\ref{fig:figSim4}, Fig.~\ref{fig:upperBound} illustrates the curves of the upper bound delineated in Theorem~\ref{thm: upper bound} and the error incurred by the
neural controller obtained using the gradient in Section
\ref{sec:grad} (both curves are plotted	 by averaging the results over the same $100$ experiments as demonstrated in Fig.~\ref{fig:figSim4}).

\begin{figure}
  \centering
  \definecolor{mycolor1}{rgb}{0.00000,0.44700,0.74100}%
\definecolor{mycolor2}{rgb}{0.85000,0.32500,0.09800}%
  \begin{tikzpicture}

\begin{axis}[%
width=0.8\columnwidth ,
height=1.0in,
scale only axis,
grid=both,
grid style={line width=.1pt, draw=gray!10},
xmin=3,
xmax=10,
xlabel={$M$ (\# of systems to be approximated)},
ymode=log,
ymin=0.1,
ymax=197500,
yminorticks=true,
ytick distance=10^2,
axis background/.style={fill=white},
legend style={at={(0.97,0.03)}, anchor=south east, legend cell align=left, align=left, draw=white!15!black,font=\small},
font = \small]
\addplot [color=mycolor2]
  table[row sep=crcr]{%
3	0.770046409013396\\
4	8.42575964818023\\
5	53.9848150787107\\
6	66.1468065503646\\
7	102.300970043218\\
8	122.463655276995\\
9	145.42432992666\\
10	235.789664375822\\
};
\addlegendentry{cost neural controller}

\addplot [color=mycolor1]
  table[row sep=crcr]{%
3	581.539691929514\\
4	867.500432669749\\
5	2018.20182587778\\
6	9179.05519446779\\
7	40417.9897348082\\
8	46374.5766607261\\
9	87582.5294219589\\
10	179100.49820345\\
};
\addlegendentry{upper bound of Theorem~\ref{thm: upper bound}}
\end{axis}
\end{tikzpicture}
    \caption{The figure illustrates two distinct curves: the red one represents the evolution of the minimum of the cost function~\eqref{eq:problem2} derived via a gradient descent algorithm that utilizes the gradient discussed in Section~\ref{sec:grad}; the blue curve shows the upper bound~\eqref{eq:upperBound}. These outcomes are averaged across the same $100$ simulations of Fig.~\ref{fig:figSim4}.}\label{fig:upperBound}
\end{figure}
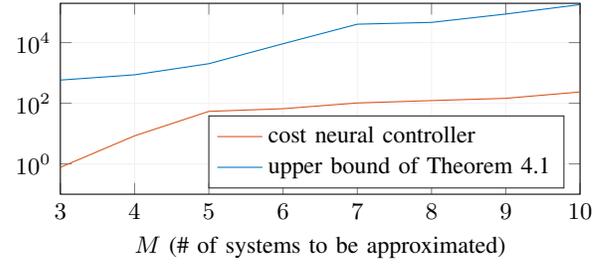

\subsection{Lower bounds}
Computing a lower bound for the multi-task control problem
\eqref{eq:problem2} presents considerable challenges, as the existing
model reduction tools cannot be applied in a straightforward way. In
this section we consider two alternative formulations of the
minimization problem \eqref{eq:problem2}, which capture the
fundamental limitations of multi-task control problems, although using
different performance metrics than in \eqref{eq:problem2}. In
particular, we consider the following multi-task control minimization
problems:
\begin{equation}\label{eq:easyLowerBound1}
\min_{\substack{W,D_1, \dots,D_M\\ \,B ,\, C}}
\underbrace{\sum_{i=1}^{M}\sup_{t \geq 0}
  \norm{\giD\left(t\right)-\giL\left(t\right) 
  }_{2}}_{\Jcal_2},
\end{equation}
and
\begin{equation}\label{eq:sumLowerBound}
\inf_{\substack{W,D_1, \dots,D_M\\ \,B ,\, C}}\underbrace{\sum_{i=1}^{M}\norm{\SiD-\SiL }_{1}}_{\Jcal_1}.
\end{equation}
In \eqref{eq:easyLowerBound1} $\norm{\cdot}_{2}$
denotes the $2$-induced matrix norm, with $\giD\left(t\right)$ and
$\giL\left(t\right)$ representing the impulse responses of $\SiD$ and
$\SiL$, respectively. For this problem we provide a simple lower
bound, which may be conservative in some~cases.
Instead, in \eqref{eq:sumLowerBound}, $\norm{\cdot}_1$ denotes the $1$-norm of
the impulse response of the system and it can be interpreted as the
induced norm of the system for signals of bounded magnitude. We will
solve this problem exactly, but only for a class of scalar
systems. We now proceed with a lower bound for \eqref{eq:easyLowerBound1}.
\begin{theorem}{\textbf{\emph{(Lower bound
        of~\eqref{eq:easyLowerBound1})}}}\label{thm: lower bound 1}
  Using the notation in \eqref{eq:linSys} and \eqref{eq:desSystems},
  we have
  \begin{equation}\label{eq:easyLowerBound}
    \min_{\substack{W,D_1, \dots,D_M,\\ \,B ,\, C}} \Jcal_2\geq
    \sqrt{\frac{\sum_{i=1}^{M}{\norm{C_iB_i-\left(\sum_{i=1}^{M}C_iB_i\right)\middle/M}_F^2}}{\min\left\lbrace
          p,\,m \right\rbrace}}.
  \end{equation}
\end{theorem}

Theorem \ref{thm: lower bound 1} provides a bound for the minimization
problem \eqref{eq:easyLowerBound1} by substituting the supremum over
time with the evaluation of the impulse response at time zero. When
doing so, the error only depends on the input and output matrices, and
it is minimized by choosing the input and output matrices of the
neural controller as the average of the input and output matrices of
the systems to be approximated. Clearly, this can result in a
conservative bound. We now prove Theorem \ref{thm: lower
  bound 1}.

\begin{proof}
We start our discussion by stating:
\begin{equation}
\label{eq:minProblemNewObjVar}
\min_{\substack{W,D_1, \dots,D_M,\\ \,B ,\, C}} \Jcal_{2}\geq \min_{\substack{X_1, \dots,X_M,\\ \,B ,\, C}} \underbrace{\sum_{i=1}^{M}\sup_{t \geq 0}\norm{\giD\left(t\right)-\giLH\left(t\right)
  }_{2}}_{ \widehat{\Jcal}_{2}}
\end{equation}
where $\giLH$ denotes the impulse response of systems $\SiLH$ defined as:
\begin{equation*}
  \SiLH = \left(
  X_i,\,B,\,C\right) ,
\end{equation*}
where we assume $X_i\in \RNN$ for each index $i = 1,\ldots,\,M$. This formulation of the problem acts as a lower bound to~\eqref{eq:problem2} since the set of optimization variables in~\eqref{eq:minProblemNewObjVar} encapsulates all optimization variables considered in~\eqref{eq:problem2}. Moreover, by setting $CB =\left(\sum_{i=1}^{M}C_iB_i\right)/M$, we obtain~\eqref{eq:easyLowerBound}. Indeed,
\begin{align}\label{eq:min2Norm}
\widehat{\Jcal}_{2} &= \sum_{i=1}^{M}\sup_{t \geq 0}\norm{C_ie^{A_it}B_i-Ce^{X_it}B}_{2} \\
&\geq \sum_{i=1}^{M}\norm{C_iB_i-CB}_{2} \geq \sum_{i=1}^{M}\sqrt{\norm{C_iB_i-CB}_{2}^2}\notag\\
&\geq \sqrt{\sum_{i=1}^{M}\norm{C_iB_i-CB}_{2}^2}
\geq \sqrt{\frac{\sum_{i=1}^{M}{\norm{C_iB_i-CB}_F^2}}{\min\left\lbrace p,\,m\right\rbrace}},\notag 
\end{align}
where the final inequality in equation~\eqref{eq:min2Norm} derives from the fact that for any $X\in\Rpm$ we have $\norm{X}_{2}\geq\frac{1}{\sqrt{d}}\norm{X}_F$, where $d$ is the rank of matrix $X$~\cite{KBP-MSP:12}. Then, minimizing last term~\eqref{eq:min2Norm}  with respect to $CB$ gives~\eqref{eq:easyLowerBound}. 
\end{proof}

We now present our last lower bound, which is valid for a special
class of stable single input, single output systems.

\begin{theorem}{\textbf{\emph{(Lower bound
        of~\eqref{eq:sumLowerBound})}}}\label{thm: lower bound 2}
  Let $\SiD = (a_i,b_i,c_i)$ satisfy $a_i < 0$ and $r_i = b_i c_i > 0$
  for all $i = 1,\ldots, M$. Let, without loss of generality,
  $r_1\leq r_2\leq \ldots \leq r_M$. Let $\SiL$ as
  in~\eqref{eq:linSys} with $N = 1$. Define the function
  \begin{equation*}
    A\left(j,\,l\right) = 
    \frac{r_j}{a_j}-\frac{r_l}{a_j}\left(\frac{\ln\left(r_l/r_j\right)}{W_{-1}\left(-1/2e\right)}
      + 1 \right),
  \end{equation*}
  valid for any pair of indices $1\leq j <l \leq M$, where
  $W_{-1}\left(\cdot\right)$ denotes the W Lambert function on the
  negative branch \cite{RMC-GHG-DEGH-DJJ-DEK:96}. Then,
  \begin{equation}\label{eq:finalMin}
    \inf_{\substack{W,D_1, \dots,D_M\\ \,B ,\, C}}\Jcal_1\geq\min \left\lbrace
      A\left(j,\,l\right),\,-A\left(l,\,j\right),-\frac{r_j}{a_j}\right
    \rbrace.
  \end{equation}
\end{theorem}

Some comments are in order. First, when $M = 2$, the bound in
\eqref{eq:finalMin} holds with equality, thus providing an optimal
solution to \eqref{eq:sumLowerBound}. In this case, it can be shown
that the optimal solution requires the neural controller to satisfy
$\norm{\SoneD-\SoneL }_{1} = 0$ or $\norm{\StwoD-\StwoL }_{1} = 0$,
that is, to equal one of the two systems to be approximated. Second,
when $M > 2$, the bound is obtained by selecting only two of the
systems to be approximated. Thus, the bound \eqref{eq:finalMin} can be
sharpened by maximizing over the indices $i$ and $j$ that correspond
to the selected systems. We now prove Theorem~\ref{thm: lower bound
  2}.

\begin{proof}
For any pair of indices $1\leq j < l \leq M$, it holds that:
\begin{equation}
\label{eq:lowerBoundTwo}
\inf_{\substack{W,D_1, \dots,D_M\\ \,B ,\, C}}
\Jcal_1 \geq \inf_{\substack{W,D_j ,D_l\\ \,B ,\, C}}\:\underbrace{\overbrace{\norm{\SjD-\SjL }_{1}}^{\widetilde{A}_j}+\overbrace{\norm{\SlD-\SlL }_{1}}^{\widetilde{A}^l}}_{\widetilde{\Jcal}_1}.
\end{equation}
Because $N=1$ we  can define the scalar quantities $r = BC$, $x_l = -1 + D_lW$ and $x_j = -1 + D_jW$ and noticing that~\eqref{eq:lowerBoundTwo} is equivalent to minimize over $r$, $x_j$ and $x_l$, we will show that the optimal value of the right-hand side of~\eqref{eq:lowerBoundTwo} are realized for:
\begin{equation}
\label{eq:tabCases}
\arraycolsep=1.6pt\def\arraystretch{1.25}
\begin{array}{l|ccc}
&\text{Case 1} & \text{Case 2} & \text{Case 3}\\
\hline
 r & r_j & r_l & r_l \\
x_l& x^*_l\left(r_j\right)&a_l & a_l\\
x_j & a_j & x^*_j\left(r_l\right)&-\infty
\end{array},
\end{equation}
where each case corresponds to a unique term within the minimum argument of inequality~\eqref{eq:finalMin}, with \begin{align}
x^*_j\left(r\,\right)&= \begin{cases}
a_j\frac{y_j^*\left(r\,\right)}{y_j^*\left(r\,\right)+1} &\text{ if }r > r_j\\
a_j &\text{ if }r = r_j
\end{cases} \label{eq:xjStar}\\
x^*_l\left(r\,\right)&= \begin{cases}
a_l\frac{y_l^*\left(r\,\right)}{y_l^*\left(r\,\right)-1} &\text{if }r < r_l\\
a_l &\text{if }r = r_l
\end{cases},\label{eq:xlStar}
\end{align}
and 
\begin{align}
y_j^*\left(r\,\right) &= \ln\left(r/r_j\right)^{-1}\left[1+W_{-1}\left(-1/2e\right)\right]\label{eq:yjOpt}\\
y_l^*\left(r\,\right) &= \ln\left(r_l/r\,\right)^{-1}\left[1+W_{-1}\left(-1/2e\right)\right].\label{eq:ywOpt}
\end{align}
To prove~\eqref{eq:tabCases} we first fix $r_j \leq r \leq r_l$ and minimize over $x_j\left(r\right)$ and $x_l\left(r\right)$, secondly we optimize the solution over $r$. The proofs for the conditions $r \leq r_j$ and $r \geq r_l$ are not included, as they lead to analogous conclusions through similar arguments. If $x_j > a_j$ and $x_l < a_l$  it is possible to establish that the minimum of $\widetilde{\Jcal}_1$ with $r$ fixed is achieved as $x_j\to a_j^+$ and $x_l\to a_l^-$. So we proceed assuming $x_j \leq a_j$ and $x_l \geq a_l$.  In particular, if $x_j < a_j$ and $x_l > a_l$, there will be a point where the impulse responses of $\SjL$ and $\SjD$, as well as $\SlL$ and $\SlD$, intersect. Applying the following changes of variable
\begin{equation*}
\begin{array}{rl}
 y_j\colon \left(-\infty,a_j\right) & \longrightarrow \left(-\infty,\,-1\right)  \\
  x_j & \longmapsto \frac{x_j}{a_j-x_j}\\
  y_l\colon \left(a_l,\,0\right) & \longrightarrow \left(-\infty,\,0\right)\\
  x_l & \longmapsto \frac{x_l}{x_l-a_l}.
\end{array}
\end{equation*}
to $\widetilde{\Jcal}_1 = \widetilde{A}_j + \widetilde{A}^l$ yields:
\begin{align}
\widetilde{A}_j\left(y_j,r\right) &= \frac{r\left(2e^{y_j \ln\left(r/r_j\right)}-1\right)}{a_j y_j}	-\frac{r}{a_j}+\frac{r_j}{a_j}\label{eq:min1}\\
\widetilde{A}^l\left(y_l,r\right) &=\frac{r\left(2e^{y_l \ln\left(r_l/r\right)}-1\right)}{a_ly_l}	+\frac{r}{a_l}-\frac{r_l}{a_l}.\notag
\end{align}
Computing $\frac{\partial}{\partial y_j}\widetilde{A}_j\left(y_j,r\right) = 0$ we derive the following equation:
\begin{equation}\label{eq:trasEq}
2 \ln\left(r/r_j\right) e^{y_j \ln\left(r/r_j\right)}-2e^{y_j \ln\left(r/r_j\right)} + 1 = 0.
\end{equation}
This transcendental equation, solvable via the Lambert W function $W_{-1}\left(\cdot\right)$ yields only one negative solution~\eqref{eq:yjOpt} for $j$, and similarly~\eqref{eq:ywOpt} for $l$. Finally, one can observe that they are also actual minima within $y_j$ and $y_l$ in the interval $\left(-\infty,\,0\right)$. For~\eqref{eq:min1} we have to distinguish two scenarios depending on the value of $r$ in $r_j \leq r \leq r_l$:
\begin{enumerate}[a)]
\item $y_j^*\in \left(-\infty,\,-1\right)$ if $r<r_je^{-\left[1+W_{-1}\left(-1/2e\right)\right ]}$\label{enum:scenarioa};
\item $y_j^* \notin \left(-\infty,\,-1\right)$ if $r\geq r_je^{-\left[1+W_{-1}\left(-1/
2e\right)\right]}$\label{enum:scenariob}.
\end{enumerate}

In particular if $r_l< r_je^{-\left[1+W_{-1}\left(-1/2e\right)\right]}$, then we are definitely in Scenario~\ref{enum:scenarioa}, and it is possible to prove that $\widetilde{\Jcal}_1$ in the interval $r_j<r< r_l$, with $x_j = x^*_j\left(r\right)$, $x_l = x^*_l\left(r\right)$,  as in \eqref{eq:xjStar} and \eqref{eq:xlStar}, admits no local minimum in $r$. Therefore, according to the Weierstrass Extreme Value Theorem, the minimum for  $r_j\leq r\leq r_l$ must be located at the boundaries of the given interval. Resulting in
\begin{equation}
\label{eq:minimum1}
\min\left\lbrace \widetilde{A}^l\left(x^*_l\left(r_j\right),r_j\right),\,\widetilde{A}_j\left(x^*_j\left(r_l\right),r_l\right)\right\rbrace = \inf_{x_j,\,x_l,\,r}\widetilde{\Jcal_1},
\end{equation}
and falling back to Cases 1 and 2 in the table~\eqref{eq:tabCases}.

Conversely, in Scenario~\ref{enum:scenariob}, when $ r_je^{-\left[1+W_{-1}\left(-1/2e\right)\right]}\leq r\leq r_l $, 
for reasons of monotonicity it can be shown that:
\begin{equation*}
\widetilde{A}_j\left(x_j,r\,\right) \xrightarrow[x_j \to -\infty]{} -\frac{r_j}{a_j} = \inf_{x_j} \widetilde{A}_j\left(x_j,r\,\right).
\end{equation*}
Since this last infimum is independent of $x_l$ and $r$, we find
\begin{align*}
\inf_{x_j,\,x_l,\,r}\widetilde{\Jcal_1}\geq -\frac{r_j}{a_j} + \underbrace{\inf_{x_l,r}\widetilde{A}^l\left(x_l,r\,\right)}_{=0}.
\end{align*}
This inequality is then related to the parameters defined in Case 3 of Table~\eqref{eq:tabCases}.

For $ r_j \leq r < r_je^{-\left[1+W_{-1}\left(-1/2e\right)\right]} $, we can proceed as before, leading to
\begin{equation}
\label{eq:minimum2}
\inf_{x_j,\,x_l,\,r}\widetilde{\Jcal_1}=\inf_{\substack{W,D_j ,D_l\\ \,B ,\, C}}\widetilde{\Jcal_1}\geq\min\left\lbrace \widetilde{A}^l\left(x^*_l\left(r_j\right),r_j\right),\,-\frac{r_j}{a_j}\right\rbrace.
\end{equation}
Combining the minima from~\eqref{eq:minimum1} and~\eqref{eq:minimum2}, we obtain the relationship~\eqref{eq:finalMin}, where $\widetilde{A}^l\left(x^*_l\left(r_j\right),r_j\right) = -A\left(l,j\right)$ and $\widetilde{A}_j\left(x^*_j\left(r_l\right),r_l\right) = A\left(j,l\right)$.
\end{proof}

\section{Conclusion and future work}\label{sec: conclusion}
This paper addresses the problem of approximating multiple linear
systems using a single non-linear neural controller. Key contributions
of this work include a characterization of the approximation
performance of the neural controller, in terms of analytical lower and
upper bounds, and the design of gradient-based algorithms to train the
controller parameters. Directions of future work include the design of
switching mechanisms to engage different controller modalities, as
well a study of the approximation properties of the neural controller
away from the pre-specified systems.



\bibliographystyle{IEEEtran}

\bibliography{./BIB/New,./BIB/Main,./BIB/FP}

\begin{thebibliography}{10}
\providecommand{\url}[1]{#1}
\csname url@samestyle\endcsname
\providecommand{\newblock}{\relax}
\providecommand{\bibinfo}[2]{#2}
\providecommand{\BIBentrySTDinterwordspacing}{\spaceskip=0pt\relax}
\providecommand{\BIBentryALTinterwordstretchfactor}{4}
\providecommand{\BIBentryALTinterwordspacing}{\spaceskip=\fontdimen2\font plus
\BIBentryALTinterwordstretchfactor\fontdimen3\font minus
  \fontdimen4\font\relax}
\providecommand{\BIBforeignlanguage}[2]{{%
\expandafter\ifx\csname l@#1\endcsname\relax
\typeout{** WARNING: IEEEtran.bst: No hyphenation pattern has been}%
\typeout{** loaded for the language `#1'. Using the pattern for}%
\typeout{** the default language instead.}%
\else
\language=\csname l@#1\endcsname
\fi
#2}}
\providecommand{\BIBdecl}{\relax}
\BIBdecl

\bibitem{BJ:2000}
J.~Baxter, ``A model of inductive bias learning,'' \emph{Journal of Artificial
  Intelligence Research}, vol.~12, pp. 149--198, 2000.

\bibitem{RC:97}
R.~Caruana, ``Multitask learning,'' \emph{Machine Learning}, vol.~28, no.~1,
  p.~41, 1997.

\bibitem{LG-FP-TP-SC:23}
L.~Gong, F.~Pasqualetti, T.~Papouin, and S.~Ching, ``Astrocytes as a mechanism
  for meta-plasticity and contextually-guided network function,'' 2024.

\bibitem{AM-SL:18}
A.~Mallya and S.~Lazebnik, ``Packnet: Adding multiple tasks to a single network
  by iterative pruning,'' in \emph{{IEEE} Conf.\ on Computer Vision and Pattern
  Recognition}, Los Alamitos, CA, USA, jun 2018, pp. 7765--7773.

\bibitem{YY-XS-CH-JZ:19}
Y.~Yu, X.~Si, C.~Hu, and J.~Zhang, ``A review of recurrent neural networks:
  Lstm cells and network architectures,'' \emph{Neural computation}, vol.~31,
  no.~7, pp. 1235--1270, 2019.

\bibitem{PJW:89}
P.~J. Werbos, ``Neural networks for control and system identification,'' Tampa,
  FL, USA, Dec. 1989, pp. 260--265.

\bibitem{BDA-JBM:07}
B.~D. Anderson and J.~B. Moore, \emph{Optimal control: linear quadratic
  methods}.\hskip 1em plus 0.5em minus 0.4em\relax Courier Corporation, 2007.

\bibitem{BK-MC:16}
B.~Kouvaritakis and M.~Cannon, ``Model predictive control,'' \emph{Switzerland:
  Springer International Publishing}, vol.~38, pp. 13--56, 2016.

\bibitem{JPH-DL-ASM:03}
J.~P. Hespanha, D.~Liberzon, and A.~S. Morse, ``Overcoming the limitations of
  adaptive control by means of logic-based switching,'' \emph{Systems \&
  Control Letters}, vol.~49, no.~1, pp. 49--65, 2003.

\bibitem{KJA:91}
K.~J. {\AA}str{\"o}m, ``Adaptive control,'' in \emph{Mathematical System
  Theory: The Influence of R. E. Kalman}, A.~C. Antoulas, Ed.\hskip 1em plus
  0.5em minus 0.4em\relax Springer Berlin Heidelberg, 1991, pp. 437--450.

\bibitem{BM:85}
B.~Mårtensson, ``The order of any stabilizing regulator is sufficient a priori
  information for adaptive stabilization,'' \emph{Systems \& Control Letters},
  vol.~6, no.~2, pp. 87--91, 1985.

\bibitem{MF-BB:86}
M.~Fu and B.~Barmish, ``Adaptive stabilization of linear systems via switching
  control,'' \emph{IEEE Transactions on Automatic Control}, vol.~31, no.~12,
  pp. 1097--1103, 1986.

\bibitem{DEM-EJD:91}
D.~Miller and E.~Davison, ``An adaptive controller which provides an
  arbitrarily good transient and steady-state response,'' \emph{IEEE
  Transactions on Automatic Control}, vol.~36, no.~1, pp. 68--81, 1991.

\bibitem{LX-LY-GC-SS:22}
L.~Xin, L.~Ye, G.~Chiu, and S.~Sundaram, ``Identifying the dynamics of a system
  by leveraging data from similar systems,'' Atlanta, GA, USA, Jun. 2022, pp.
  818--824.

\bibitem{YC-AMO-FP-EDA:23}
Y.~Chen, A.~M. Ospina, F.~Pasqualetti, and E.~Dall'Anese, ``Multi-task system
  identification of similar linear time-invariant dynamical systems,'' Marina
  Bay Sands, Singapore, Dec. 2023, to appear. arXiv preprint arXiv:2301.01430.

\bibitem{TTZ-KK-BDL:22}
T.~T. Zhang, K.~Kang, B.~D. Lee, C.~Tomlin, S.~Levine, S.~Tu, and N.~Matni,
  ``Multi-task imitation learning for linear dynamical systems,'' \emph{arXiv
  preprint arXiv:2212.00186}, 2022.

\bibitem{TG-AAAM-VK-FP:23}
T.~Guo, A.~A. {Al Makdah}, V.~Krishnan, and F.~Pasqualetti, ``Imitation and
  transfer learning for {LQG} control,'' vol.~7, pp. 2149--2154, 2023.

\bibitem{BR-DT-AZ:22}
B.~Richards, D.~Tsao, and A.~Zador, ``The application of artificial
  intelligence to biology and neuroscience,'' \emph{Cell}, vol. 185, no.~15,
  pp. 2640--2643, 2022.

\bibitem{GIP-RK-JLP-CK-SW:19}
G.~I. Parisi, R.~Kemker, J.~L. Part, C.~Kanan, and S.~Wermter, ``Continual
  lifelong learning with neural networks: A review,'' \emph{Neural Networks},
  vol. 113, pp. 54--71, 2019.

\bibitem{SM-AO-LA-PJ:2014}
M.~Soare, O.~Alsharif, A.~Lazaric, and J.~Pineau, ``Multi-task linear
  bandits,'' in \emph{NIPS2014 Workshop on Transfer and Multi-task Learning:
  Theory meets Practice}, 2014.

\bibitem{AAD-UD-CS:2017}
A.~A. Deshmukh, U.~Dogan, and C.~Scott, ``Multi-task learning for contextual
  bandits,'' in \emph{Advances in Neural Information Processing Systems},
  vol.~30.\hskip 1em plus 0.5em minus 0.4em\relax Curran Associates, Inc.,
  2017, pp. 4851--4859.

\bibitem{DK-MB-MFA:96}
D.~Kavrano{\u{g}}lu, M.~Bettayeb, and M.~F. Anjum, ``$\mathcal{L}_{\infty}$
  norm simultaneous system approximation,'' vol.~6, no. 9-10, pp. 999--1014,
  1996.

\bibitem{SGG-NS-LP-JFO:18}
S.~Guerra-Gomes, N.~Sousa, L.~Pinto, and J.~F. Oliveira, ``Functional roles of
  astrocyte calcium elevations: From synapses to behavior,'' \emph{Frontiers in
  Cellular Neuroscience}, vol.~11, p. 427, 2018.

\bibitem{CHTT-GP-GRG:18}
C.~H.~T. Tran, G.~Peringod, and G.~R. Gordon, ``Astrocytes integrate behavioral
  state and vascular signals during functional hyperemia,'' \emph{Neuron}, vol.
  100, no.~5, pp. 1133--1148.e3, 2018.

\bibitem{CMR-SC-TP:23}
C.~Murphy-Royal, S.~Ching, and T.~Papouin, ``A conceptual framework for
  astrocyte function,'' vol.~26, no.~11, pp. 1848--1856, 2023.

\bibitem{JV-BV-WM-SV-MD:09}
J.~Vanbiervliet, B.~Vandereycken, W.~Michiels, S.~Vandewalle, and M.~Diehl,
  ``The smoothed spectral abscissa for robust stability optimization,''
  vol.~20, no.~1, pp. 156--171, 2009.

\bibitem{KBP-MSP:12}
K.~B. Petersen and M.~S. Pedersen, \emph{The Matrix Cookbook}.\hskip 1em plus
  0.5em minus 0.4em\relax Technical University of Denmark, 2012.

\bibitem{KJA-RMM:08}
K.~J. {\r A}str\"om and R.~M. Murray, \emph{Feedback Systems}, 2008.

\bibitem{MST-IP:87}
M.~S. Tombs and I.~Postlethwaite, ``Truncated balanced realization of a stable
  non-minimal state-space system,'' \emph{International Journal of Control},
  vol.~46, pp. 1319--1330, 1987.

\bibitem{AAC-DCS:01}
A.~C. Antoulas and D.~C. Sorensen, ``Approximation of large-scale dynamical
  systems: an overview,'' \emph{Int. J. Applied Mathematics and Computer
  Science}, vol.~11, no.~5, pp. 1093--1121, 2001.

\bibitem{RMC-GHG-DEGH-DJJ-DEK:96}
R.~M. Corless, G.~H. Gonnet, D.~E.~G. Hare, D.~J. Jeffrey, and D.~E. Knuth,
  ``On the lambert w function,'' \emph{Advances in Computational Mathematics},
  vol.~5, no.~1, pp. 329--359, 1996.

\end{thebibliography}

\end{document}